\documentclass[11pt]{article}

\usepackage[margin=1in]{geometry}
\usepackage{authblk}
\usepackage[hidelinks]{hyperref}

\usepackage{algorithm2e}
\usepackage{amsmath}
\usepackage{amsthm}
\usepackage{amsfonts}
\usepackage{amssymb}
\usepackage{graphicx}
\usepackage{microtype}

\usepackage{float}      
\usepackage{flafter}    
\usepackage[section]{placeins} 

\usepackage{tikz}
\usetikzlibrary{arrows.meta,positioning}

\theoremstyle{plain}
\newtheorem{theorem}{Theorem}
\newtheorem{lemma}{Lemma}
\newtheorem{proposition}{Proposition}
\newtheorem{corollary}{Corollary}
\theoremstyle{definition}
\newtheorem{definition}{Definition}
\theoremstyle{remark}
\newtheorem{remark}{Remark}

\begin{document}

\title{Stochastic Indexing Primitives for Non-Deterministic Molecular Archives}

\author[1]{Faruk Alpay}
\author[1]{Levent Sarioglu}

\affil[1]{Bahcesehir University, Department of Computer Engineering, Istanbul, Turkey}

\affil[1]{\href{mailto:faruk.alpay@bahcesehir.edu.tr}{faruk.alpay@bahcesehir.edu.tr}, \href{mailto:levent.sarioglu@bahcesehir.edu.tr}{levent.sarioglu@bahcesehir.edu.tr}}

\date{}

\maketitle

\begin{abstract}
Random access remains a central bottleneck in DNA-based data storage: existing systems typically retrieve records by PCR enrichment or other multi-step biochemical procedures, which do not naturally support fast, massively parallel, content-addressable queries.

This paper introduces the \emph{Holographic Bloom Filter} (HBF), a new probabilistic indexing primitive that encodes key--pointer associations as a single high-dimensional memory vector. HBF binds a key vector and a value (pointer) vector via circular convolution and superposes bindings across all stored records; queries decode by correlating the memory with the query key and selecting the best-matching value under a margin-based decision rule.

We provide algorithms for construction and decoding, and develop a detailed probabilistic analysis under explicit noise models (memory corruption and query/key mismatches). The analysis yields concentration bounds for match and non-match score distributions, explicit threshold/margin settings for a top-$K$ decoder, and exponential error decay in the vector dimension under standard randomness assumptions.

Overall, HBF offers a concrete, analyzable alternative to pointer-chasing molecular data structures, enabling one-shot associative retrieval while explicitly quantifying the trade-offs among dimensionality, dataset size, and noise.
\end{abstract}


\section{Introduction}
With the rapid advancement of synthesis and sequencing technologies, DNA macromolecules have emerged as a promising medium for long-term digital information storage, primarily owing to their exceptional density and durability \cite{church2012,goldman2013}. However, beyond encoding and error correction, scalable \emph{indexing and random access} remain open algorithmic problems at the molecular layer.

\paragraph{Problem statement.}
We study a molecular archive containing $n$ records with unique keys $x_i$ and associated pointers/labels $y_i\in\mathcal{Y}$. The goal is to design an \emph{index representation} $M$ and a \emph{decoder} $\mathsf{Dec}$ such that, for a noisy query key $\tilde{x}$,
\[
\hat{y} \;=\; \mathsf{Dec}(\tilde{x}; M)\in\mathcal{Y}\cup\{\bot\}
\]
returns the correct $y_i$ with high probability, while requiring a constant number of logical ``rounds'' of molecular interaction. Formally, we aim to maximize
\[
\Pr\big\{\mathsf{Dec}(\tilde{x}_i;M)=y_i\big\} \quad\text{subject to a one-shot query model.}
\]

\paragraph{Contributions.}
This paper proposes the \emph{Holographic Bloom Filter (HBF)} and provides (i) explicit construction and decoding algorithms (including a top-$K$ margin decoder), (ii) a signal--noise decomposition for the decoder scores, (iii) concentration and extreme-value analyses that yield closed-form threshold/margin settings, and (iv) correctness guarantees with error probabilities that decay exponentially in the vector dimension under explicit noise models.

A typical DNA storage process involves synthesizing artificial DNA strands to encode user information, storing them in a container, and subsequently sequencing them to translate the DNA back into digital sequences. These features have spurred a series of landmark demonstrations of DNA storage, encoding digital text, images, and even video into synthetic DNA and successfully recovering them via high-throughput sequencing \cite{church2012,goldman2013,erlich2017,organick2018}. However, due to technical constraints, early DNA storage architectures required reading out the entire DNA library to retrieve any specific file or data object, as there was no mechanism for direct access \cite{church2012,goldman2013}. Despite the immense potential of DNA storage, current sequencing workflows remain costly and slow compared to conventional digital storage systems \cite{shomorony2022}. This inefficiency is partly due to the need to sequence massive numbers of molecules at high coverage depth, highlighting the urgent need for more efficient retrieval techniques.

To address this challenge, recent efforts have emphasized developing \emph{random access} methods in DNA data storage, which allow specific files or subsets of data to be read without sequencing everything \cite{organick2018,imburgia2025}. One common approach is to use conventional PCR-based enrichment, wherein a unique primer corresponding to the desired file's address sequence is added to the DNA pool to selectively amplify only those strands, which are then sequenced \cite{organick2018}. While straightforward in concept, this strategy faces significant limitations. The end-to-end protocol is time-consuming and requires careful design of primers to avoid molecular crosstalk with data payloads. Furthermore, multiplexed retrieval (targeting multiple files in a single reaction) is challenging since large pools of DNA make it difficult to prevent primers from interfering with each other \cite{organick2018}. These issues worsen as database size scales up. Some designs have resorted to physically partitioning data---e.g., the DENSE system which uses nested addressing in separate compartments---or encapsulating each file in microdroplets or beads to achieve separation of file subsets. One study demonstrated fluorescent-assisted sorting (FAS) of such DNA capsules to enable boolean queries on file tags, and another method uses controlled encapsulation with magnetic beads to selectively retrieve files. These methods achieve repeated random accesses by mitigating molecular crosstalk, but they are complex to implement and impose new limitations on scalability (for instance, FAS is limited by the number of distinct fluorophores available, and encapsulation-based methods incur significant overhead in preparation time and storage density).

Given these challenges, a random access technique that provides a simple protocol with parallel (multiplexed) access capabilities is highly desirable. Furthermore, modern data retrieval requirements often extend beyond exact file lookup; one may wish to perform \emph{similarity search} or associative queries (e.g., finding data items related to a given keyword or image). Recently, a proof-of-concept content-based image search was demonstrated in a DNA storage system by leveraging DNA hybridization to cluster similar data items, though the process required a long and delicate thermal program to execute. 

The CRISPR-Cas9 system, known for its sequence-specific, programmable targeting of DNA, presents an interesting approach for molecular data retrieval. Imburgia \emph{et al.} \cite{imburgia2025} introduced a method in which Cas9 enzymes, guided by customizable RNA sequences, selectively cleave and thereby enrich desired DNA records from a large pool. This Cas9-based random access (C9RA) approach was validated on a pool of 25 files (comprising 1.6 million distinct DNA strands), demonstrating that multiple targets can be extracted in parallel using distinct guide RNAs. Moreover, they developed a \emph{similarity search} mode by exploiting Cas9’s tolerance for mismatches: a guide RNA encoding a query (e.g., an image’s feature vector) will bind not only exact matches but also sequences with small differences, thereby retrieving a set of DNA records with related content. While powerful, such Cas9-driven retrieval still requires complex biological workflows (protein delivery, reaction optimization) and currently achieves readout latencies on the order of hours rather than seconds.

In light of these developments, we seek to develop a \textbf{theoretically grounded, computer science-centric solution} for indexing data within a DNA archive. We aim to enable \emph{in situ} queries---queries answered within the molecular substrate---without exhaustive sequencing or extensive lab procedures for each query. We identify three key requirements for such an indexing primitive: (1) It must support \textbf{direct content-addressable queries} (analogous to $O(1)$ key-value lookups) to avoid sequentially scanning huge DNA pools. (2) It should be \textbf{robust to errors} at multiple levels: individual nucleotides may be incorrectly synthesized or read, and query probes may bind imperfectly, so the index must gracefully handle approximate matches and noise. (3) It must operate within the \textbf{massively parallel, biochemical environment} of DNA, leveraging the fact that a vast number of molecules can interact simultaneously while minimizing the number of sequential wet-lab steps (each of which adds significant time and energy costs).

To meet these requirements, we propose the \emph{Holographic Bloom Filter (HBF)}, a novel indexing data structure that marries the classic Bloom filter’s probabilistic membership testing \cite{bloom1970} with hyperdimensional computing techniques \cite{kanerva2009,plate1995} to enable approximate, associative retrieval. In essence, the HBF treats the entire DNA pool as a high-dimensional "memory vector" that can be queried in parallel by a suitably encoded molecular probe. The term \emph{holographic} refers to the fact that each data item is not stored in a single physical location (as a pointer in a tree would be), but rather as a distributed pattern spread across many molecules or nucleotides of a composite representation. This concept is inspired by \emph{Vector Symbolic Architectures} (VSA) \cite{kanerva2009,plate1995}, which use high-dimensional vectors to represent items and employ operations like circular convolution to bind and associate those vectors. By applying VSA operations to DNA-encoded vectors, our HBF creates a dense index that can be queried via similarity operations, rather than by following a chain of pointers.

Our approach is \textbf{algorithmic and theoretical}---we abstract the DNA storage medium as a noisy memory where certain operations (like parallel associative search via hybridization) are available, but we do not assume a specific laboratory protocol. This allows us to focus on designing and analyzing the data structure itself. The main contributions of this work are as follows:
\begin{itemize}
    \item We introduce the \textbf{Holographic Bloom Filter (HBF)}, a probabilistic indexing structure for DNA data storage that enables direct random access and associative queries. The HBF uses circular convolution on high-dimensional vectors (derived from DNA $k$-mer embeddings) to simulate multiple independent hash functions, encoding key-pointer associations in a single composite memory.
    \item We provide a rigorous \textbf{theoretical analysis} of the HBF’s performance and noise robustness. We derive formulas for the false positive rate of HBF in the presence of random noise, and we prove two theorems that bound the probability of spurious retrievals (false positives) when query vectors differ from stored keys by a given Hamming distance. We also analyze trade-offs between vector dimensionality, number of stored items, and error rates, drawing parallels to classical Bloom filter theory.
    \item We compare HBF to a \textbf{Stochastic Skip-Graph} baseline (a skip-list-like probabilistic pointer structure adapted for DNA archives). In a skip-graph \cite{aspnes2003}, each lookup requires $O(\log n)$ pointer traversals; in a molecular context, this would entail iterative rounds of chemical reactions, compounding error probabilities and latency at each step. By contrast, HBF executes a \emph{single-step} query in which all potential matches are evaluated in parallel, effectively achieving $O(1)$ lookup time on average. We argue that this yields significantly lower retrieval latency and energy consumption for large-scale molecular archives.
    \item We present \textbf{pseudocode algorithms} for constructing and querying the HBF, and discuss how parameters can be tuned for different error regimes. We also include conceptual figures illustrating the HBF’s architecture and a complexity comparison against baseline approaches.
\end{itemize}

The rest of the paper is organized as follows. Section~\ref{sec:background} provides background on DNA storage, Bloom filters, skip graphs, and VSA relevant to our approach. Section~\ref{sec:hbf-design} describes the HBF design in detail, including how DNA sequences are encoded as vectors and how insertions and queries are performed. Section~\ref{sec:analysis} presents a theoretical analysis of the HBF’s reliability, with formal theorems on false positive and false negative probabilities under realistic error models. Section~\ref{sec:experiments} discusses performance considerations and compares our scheme to a skip-graph baseline. Finally, Section~\ref{sec:conclusion} concludes with a summary and potential directions for future research.

\section{Model and Related Work}
\label{sec:background}
\subsection{Formal Model (Indexing as Noisy Associative Decoding)}
We abstract the molecular substrate as a noisy channel acting on an index state $M\in\mathbb{R}^d$ and a query probe derived from a key.
Let $\mathbf{k}_x,\mathbf{v}_y\in\{\pm1\}^d$ denote the (pseudo)random embeddings of keys and values.
The HBF memory is a superposition
\[
M \;=\; \sum_{i=1}^n \alpha_i\,(\mathbf{k}_{x_i}*\mathbf{v}_{y_i}),
\]
with weights $\alpha_i$ (often $\alpha_i\equiv \rho/\sqrt{n}$).
A noisy stored memory is modeled as
\[
\tilde{M} \;=\; M + W,\qquad W\sim \mathcal{N}(0,\sigma_M^2 I_d)\ \text{(or componentwise flips/perturbations)}.
\]
A noisy query key is modeled by a perturbed embedding
\[
\mathbf{k}_{\tilde{x}} \;=\; \mathbf{k}_x + E,\qquad E\sim \mathcal{N}(0,\sigma_Q^2 I_d)\ \text{(or Hamming-distance $H$ flips)}.
\]
Decoding computes $\mathbf{z}=\tilde{M}\circledast \mathbf{k}_{\tilde{x}}$ and selects
\[
\hat{y} \;\in\; \arg\max_{y\in\mathcal{Y}}\ \langle \mathbf{z},\mathbf{v}_y\rangle
\]
subject to threshold/margin constraints. This is a high-dimensional hypothesis test with structured interference, and the subsequent sections focus on deriving explicit, analyzable thresholds.

\subsection{DNA Data Storage and Random Access (Context)}
In DNA storage systems, data is encoded as synthetic DNA oligonucleotides (short sequences of bases, typically 100--200 bp) that collectively represent the digital information. To recover the data, the DNA pool is sequenced and decoded back into binary. Early demonstrations by Church \emph{et al.} and Goldman \emph{et al.} simply sequenced the entire pool and then reconstructed the stored files \cite{church2012,goldman2013}. This approach does not scale for large archives, as it requires enormous sequencing effort even to retrieve a single file. To enable random access, later designs introduced the concept of attaching a unique \emph{address sequence} or barcode to each data strand. By including a short key sequence for each file, one can in principle retrieve that file by an operation that targets its key.

Organick \emph{et al.} (2018) implemented this idea at scale by designing a library of orthogonal primers for 35 distinct files (over 200~MB of data in total) and demonstrating that each file could be individually retrieved via PCR enrichment \cite{organick2018}. This was a significant step forward, proving that random access is feasible in DNA storage. However, as the number of files $n$ grows into the thousands or millions, designing and managing unique primers for each becomes impractical. Primers can exhibit non-specific binding or form primer-dimers, and in a complex pool, amplification biases can cause some targets to dominate while others drop out. Indeed, Organick \emph{et al.} reported that careful normalization of primer concentrations was needed to achieve balanced retrieval across files \cite{organick2018}. 

Other groups have sought alternatives to PCR-based selection. The CRISPR approach by Imburgia \emph{et al.} \cite{imburgia2025}, as discussed, bypasses primer design by using a programmable enzyme to directly find target sequences. This method can scale to large libraries (by programming new guide RNAs as needed), but each query still entails a separate biochemical reaction and the involvement of protein machinery.

Several systems-level DNA storage prototypes also highlight the importance of a clean, scalable addressing layer. For example, Tabatabaei Yazdi \emph{et al.} demonstrated a rewritable, random-access DNA storage architecture \cite{yazdi2015}, and Bornholt \emph{et al.} proposed an end-to-end DNA storage system with a focus on practical encode/decode and access workflows \cite{bornholt2016}. These lines of work complement our goal: rather than proposing another biochemical selection protocol, we explore an indexing abstraction that could be layered on top of such systems.

In summary, efficient random access in DNA archives remains challenging. Our work addresses this by proposing a new data structure that leverages the native parallelism of molecular interactions to perform indexing in the solution itself, rather than externally via primers or physical separation.

\subsection{Probabilistic Data Structures}
A \textbf{Bloom filter} \cite{bloom1970} is a classic probabilistic data structure for set membership queries. It consists of an array of $m$ bits initially set to 0, along with $k$ independent hash functions. To insert an item (key) into a Bloom filter, one computes its $k$ hash values (each in $0 \ldots m-1$) and sets those array positions to 1. To query membership of a key, one checks its $k$ hashed positions: if any position is 0, the key is definitely not in the set; if all are 1, the answer is "possibly in set" (with some false positive probability). By choosing $m$ and $k$ appropriately for the number of stored items $n$, one can make the false positive probability $\epsilon$ arbitrarily small. Importantly, Bloom filters have no false negatives (assuming no bit flips), but false positives occur when a non-member key happens to have all its hashed positions set to 1 by other insertions.

Bloom filters, however, only support membership testing (and union/intersection operations)---they do not retrieve associated values, nor are they designed to handle errors in the key or data (if the query key is even slightly different, e.g. a mutated sequence, a standard Bloom filter will likely return false). Our HBF will extend the Bloom filter concept to an \emph{associative} setting, returning pointers (or file identifiers) rather than a yes/no, and will be robust to partial mismatches.

On the other hand, pointer-based structures like balanced trees or skip lists provide deterministic lookup but map poorly to DNA. A \textbf{Skip List} \cite{pugh1990} maintains multiple levels of sorted linked lists to allow binary-search-like traversal in $O(\log n)$ expected time. A distributed variant, the \textbf{Skip Graph} \cite{aspnes2003}, provides similar $O(\log n)$ search functionality in peer-to-peer networks and is resilient to node failures. Conceptually, one could implement a skip list/graph in DNA by encoding pointers as complementary DNA sequences that link nodes. However, following pointers in a chemical solution is inherently sequential: each pointer hop might require a separate hybridization step and possibly a biochemical separation to remove unbound molecules before the next hop. If each hop succeeds with probability $p$ and takes time $T$, then $\ell$ hops succeed with probability $p^\ell$ and take time $\ell T$. For moderate $\ell = O(\log n)$, this can become extremely slow and unreliable as discussed earlier. 

In essence, structures that rely on following a chain of pointers (like skip lists or graphs) incur a serious performance penalty in molecular settings: unlike electronic memory, where pointer chasing is fast, in DNA each pointer resolution is a slow, probabilistic event. This motivates our focus on pointer-free, parallel query mechanisms like those embodied by HBF.

\subsection{Vector Symbolic Architectures (VSA)}
The key enabling technique for HBF is the use of \textbf{high-dimensional random vectors} to encode information, an idea that originates from \emph{Vector Symbolic Architectures} (also known as hyperdimensional computing) \cite{kanerva2009}. In a VSA, items are represented as vectors in $\{\pm 1\}^d$ for large $d$ (on the order of 10,000). Two crucial operations are defined: \emph{superposition} (addition) for combining vectors, and \emph{binding} for linking vectors together. Binding has an inverse operation that allows one to retrieve a component of a composite from the whole.

In particular, we focus on the VSA model of \emph{Holographic Reduced Representations (HRR)} introduced by Plate \cite{plate1995}, where binding is implemented via \textbf{circular convolution}. Given two vectors $A, B \in \mathbb{R}^d$, their circular convolution $C = A * B$ is a vector in $\mathbb{R}^d$. There is an inverse operation, circular correlation, such that $C \circledast A \approx B$ and $C \circledast B \approx A$ (with exact equality if the vectors are orthonormal). If $A$ and $B$ are random $\{\pm 1\}$ vectors, $C$ will also appear random, and correlation $C \circledast A$ yields a result that is close to $B$ (with noise that vanishes as $d \to \infty$).

To use an example: suppose we have a key $x$ and a value $y$ that we want to associate. We assign $x$ a random high-dimensional vector $\mathbf{k}_x$ and $y$ a random vector $\mathbf{v}_y$. The bound representation is $\mathbf{u}_{x,y} = \mathbf{k}_x * \mathbf{v}_y$. To retrieve $y$ given $x$, we compute $\mathbf{u}_{x,y} \circledast \mathbf{k}_x$, which should return (an approximation of) $\mathbf{v}_y$. If we have many pairs stored as a superposition $M = \sum_i \mathbf{k}_{x_i} * \mathbf{v}_{y_i}$, then given a query $x_q$, we compute $\mathbf{z} = M \circledast \mathbf{k}_{x_q} = \sum_i (\mathbf{k}_{x_i} * \mathbf{v}_{y_i}) \circledast \mathbf{k}_{x_q}$. If $x_q$ equals one of the stored keys, say $x_j$, then the term $i=j$ yields $\mathbf{v}_{y_j}$, while terms $i \neq j$ contribute roughly random noise (since $\mathbf{k}_{x_i} \circledast \mathbf{k}_{x_q}$ is small for $i \neq j$). Thus $\mathbf{z}$ will have a strong component along $\mathbf{v}_{y_j}$, allowing us to identify $y_j$. If $x_q$ is not in the set, $\mathbf{z}$ will be mostly noise.

This associative memory property of VSAs is what HBF leverages. We treat key DNA sequences as keys in a high-dimensional space and data location pointers as values, binding them in a superposed memory. One can view this as a high-dimensional, noise-tolerant analog of a hash table: whereas a hash table stores an index of keys to values explicitly (and requires exact matching), our HBF stores a smeared-out encoding of all key-value pairs and retrieves by approximate correlation, which inherently tolerates some noise in the input.

\subsection{Notation and Mathematical Preliminaries}
\label{sec:prelim}
We formalize the algebra used throughout the paper.

\begin{definition}[Circular convolution and correlation]
Let $a,b\in\mathbb{R}^d$ with indices modulo $d$. The \emph{circular convolution} $c=a*b\in\mathbb{R}^d$ is defined by
\[
  (a*b)[t] \;:=\; \sum_{j=0}^{d-1} a[j]\,b[t-j].
\]
The \emph{circular correlation} $r=a\circledast b\in\mathbb{R}^d$ is defined by
\[
  (a\circledast b)[t] \;:=\; \sum_{j=0}^{d-1} a[j]\,b[t+j].
\]
\end{definition}

\begin{remark}[Fast implementation]
Let $\mathcal{F}$ denote the discrete Fourier transform (DFT) on length-$d$ vectors, and let $\odot$ denote elementwise multiplication. Then
$\mathcal{F}(a*b)=\mathcal{F}(a)\odot\mathcal{F}(b)$ and $\mathcal{F}(a\circledast b)=\overline{\mathcal{F}(a)}\odot\mathcal{F}(b)$.
Hence, on conventional hardware, both operations can be implemented in $O(d\log d)$ time via FFT.
\end{remark}

\begin{definition}[Random sign model]
Unless stated otherwise, key and value vectors are drawn as i.i.d. Rademacher vectors: for $u\in\{\mathbf{k}_x,\mathbf{v}_y\}$, each coordinate satisfies $u[j]\in\{\pm1\}$ with $\Pr(u[j]=1)=\Pr(u[j]=-1)=1/2$, independently across $j$.
\end{definition}

\begin{lemma}[Inner-product concentration]
\label{lem:rademacher_inner_product}
Let $u,v\in\{\pm1\}^d$ be independent. Then $\langle u,v\rangle=\sum_{j=1}^d u[j]v[j]$ satisfies
\[\Pr\{ |\langle u,v\rangle| \ge t\} \le 2\exp\Big(-\frac{t^2}{2d}\Big)\qquad\text{for all }t\ge 0.
\]
\end{lemma}
\begin{proof}
The random variables $u[j]v[j]$ are independent Rademacher signs. The bound follows from Hoeffding's inequality for bounded, mean-zero sums.
\end{proof}

In designing HBF, we assume that each possible key and value in our system is assigned a fixed random vector (for keys, we can derive it from the DNA sequence via a hash or mapping of $k$-mers to vectors; for values, which could be file IDs or addresses, we assign random vectors as well). These assignments are known to both the encoding and query processes (like a shared dictionary). The randomness of these vectors ensures that cross-correlations behave like random noise, simplifying analysis.

\section{HBF Design and Algorithms}
\label{sec:hbf-design}
We now formalize the HBF data structure and describe algorithms for insertion and query. 

\subsection{Data Structure Definition}
Let $d$ be the dimension of the vectors used in HBF (e.g., $d = 10{,}000$). We assume two fixed families of random vectors:
\begin{itemize}
    \item For each possible key $x$ (drawn from the universe of DNA addresses), a random vector $\mathbf{k}_x \in \{\pm 1\}^d$.
    \item For each possible value/pointer $y$ (which might represent an address for data retrieval, such as a file identifier or a PCR primer sequence), a random vector $\mathbf{v}_y \in \{\pm 1\}^d$.
\end{itemize}
These can be generated by a pseudorandom function so that we do not explicitly store all vectors.

An HBF representing a set of key-value pairs $\{(x_i, y_i)\}_{i=1}^n$ is defined as:
\[ 
    M \;=\; \sum_{i=1}^n \left(\mathbf{k}_{x_i} * \mathbf{v}_{y_i}\right),
\] 
where $*$ denotes circular convolution. In words, we convolve each key's vector with its value's vector to get a binding, and then we superpose (add) all bindings to form a single memory vector $M$. The memory $M$ is itself a $d$-dimensional vector, which could be stored implicitly by DNA molecules (each position corresponding to a particular partial sequence, etc.) or explicitly in a digital simulation.

\begin{figure}[htbp]
\centering
\begin{tikzpicture}[
  every node/.style={font=\small},
  node distance=9mm,
  box/.style={draw, rounded corners, align=center, minimum height=7mm, inner sep=3pt},
  >={Latex[length=2mm]},
  line/.style={->, shorten >=2pt, shorten <=2pt}
]
  \node[box, minimum width=18mm] (k) {$\mathbf{k}_x$};
  \node[box, minimum width=18mm, right=14mm of k] (v) {$\mathbf{v}_y$};
  \node[box, minimum width=24mm, right=14mm of v] (bind) {$\mathbf{k}_x * \mathbf{v}_y$};
  \node[box, minimum width=34mm, below=12mm of bind] (mem) {Memory $M$\\$\sum_i\,\mathbf{k}_{x_i}*\mathbf{v}_{y_i}$};

  \draw[line] (k.east) -- (v.west);
  \draw[line] (v.east) -- (bind.west);
  \draw[line] (bind.south) -- (mem.north);

  \node[box, minimum width=18mm, below=12mm of k] (kq) {$\mathbf{k}_{x_q}$};
  \node[box, minimum width=34mm, below=12mm of mem] (z) {$\mathbf{z}=M\circledast\mathbf{k}_{x_q}$};
  \node[box, minimum width=16mm, right=14mm of z] (out) {$y^*$};

  \draw[line] (kq.east) -- (z.west);
  \draw[line] (mem.south) -- (z.north);
  \draw[line] (z.east) -- (out.west);
\end{tikzpicture}
\caption{High-level schematic of HBF: key/value vectors are bound by circular convolution and superposed into $M$; queries correlate $M$ with a key vector to recover the associated value.}
\label{fig:hbf_schematic}
\end{figure}
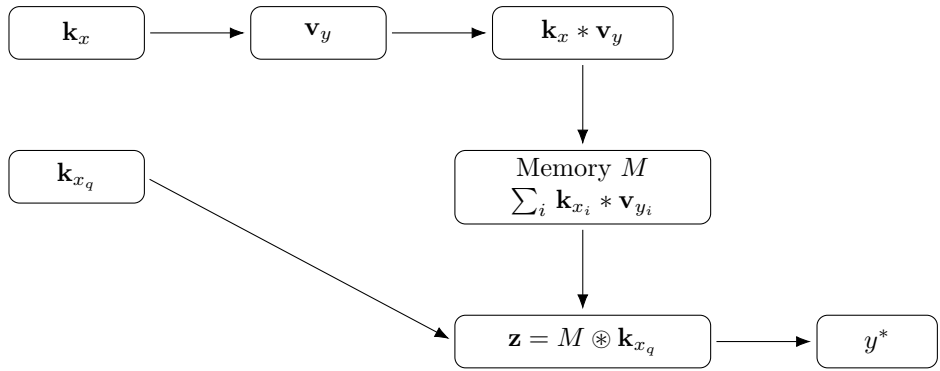

\subsection{Insertion Algorithm}
To insert a new key-value pair $(x, y)$, we need to update the memory vector $M := M + (\mathbf{k}_x * \mathbf{v}_y)$. In a DNA implementation, this could correspond to synthesizing DNA strands that encode the convolution pattern $\mathbf{k}_x * \mathbf{v}_y$ and adding them to the pool.

\subsection{Batch Construction (Practical Encoder View)}
In archival settings, the index is typically constructed offline from a list of records. We therefore also state a batch encoder that emphasizes parameterization and FFT-based binding.

\begin{algorithm}[ht]
\caption{HBF\_Build$(\mathcal{D}, d, \rho)$}
\DontPrintSemicolon
\KwIn{Dataset $\mathcal{D}=\{(x_i,y_i)\}_{i=1}^n$; dimension $d$; optional gain/normalization $\rho>0$.}
\KwOut{Memory vector $M\in\mathbb{R}^d$.}
Initialize $M\leftarrow \mathbf{0}\in\mathbb{R}^d$\;
\For{$i\leftarrow 1$ \KwTo $n$}{
  $\mathbf{k}\leftarrow \mathrm{KeyVec}(x_i)\in\{\pm1\}^d$\;
  $\mathbf{v}\leftarrow \mathrm{ValVec}(y_i)\in\{\pm1\}^d$\;
  $\mathbf{u}\leftarrow \mathrm{IFFT}(\mathrm{FFT}(\mathbf{k})\odot \mathrm{FFT}(\mathbf{v}))$ \tcp*[l]{bind: $\mathbf{k}*\mathbf{v}$}
  $M\leftarrow M + \rho\,\mathbf{u}$\;
}
\Return $M$\;
\end{algorithm}

In a molecular realization, the FFT/IFFT is not performed digitally; rather, the algorithm is a convenient mathematical abstraction of a binding operator that can be physically approximated by compositional sequence operations.

Pseudocode is given below. (Note: in archival scenarios, insertions might be done in batch offline, but we describe an abstract insert for completeness.)

\vspace{1em}
\begin{algorithm}[h]
\caption{HBF\_Insert$(x, y)$}
\DontPrintSemicolon
\KwData{Current memory vector $M \in \mathbb{R}^d$; random key vector $\mathbf{k}_x$; random value vector $\mathbf{v}_y$.}
\KwIn{Key $x$; Value/pointer $y$.}
$\mathbf{u} \leftarrow \mathbf{k}_x * \mathbf{v}_y$ \tcp*[l]{Bind key and value via circular convolution}
$M \leftarrow M + \mathbf{u}$ \tcp*[l]{Superpose this contribution into memory}
\end{algorithm}
\vspace{1em}

\subsection{Query Algorithm}
To query a key $x_q$, we compute the circular correlation of the memory with the key’s vector: $\mathbf{z} = M \circledast \mathbf{k}_{x_q}$. Ideally, $\mathbf{z}$ will equal $\mathbf{v}_y$ (the associated value vector) plus noise contributions from other stored pairs. We then find which value vector $\mathbf{v}_w$ is closest to $\mathbf{z}$ (e.g., by computing dot products). Because all $\mathbf{v}$’s are known, this can be done by parallel matching (conceptually, each possible value “probes” $\mathbf{z}$ for similarity).

To make the decoding rule explicit (and closer to a robust statistical test), we introduce a \emph{top-$K$ margin} decoder: return $y^*$ only if the gap between the best and second-best scores exceeds a margin, and the best score exceeds an absolute threshold.

We incorporate a threshold to decide if the match is strong enough or if the query should be deemed absent. If the top match score is below a set threshold, we return $\bot$ (no result). Otherwise we return the best-matching value.

Pseudocode is given below.

\vspace{1em}
\begin{algorithm}[ht]
\caption{HBF\_Decode$(x_q;\tau,\Delta,K)$}
\DontPrintSemicolon
\KwData{Memory vector $M$; key-vector generator $\mathrm{KeyVec}(\cdot)$; value codebook $\{\mathbf{v}_w\}_{w\in\mathcal{Y}}$.}
\KwIn{Query key $x_q$; absolute threshold $\tau$; margin threshold $\Delta$; top-$K$ list size.}
$\mathbf{k}\leftarrow \mathrm{KeyVec}(x_q)$\;
$\mathbf{z}\leftarrow \mathrm{IFFT}(\overline{\mathrm{FFT}(M)}\odot \mathrm{FFT}(\mathbf{k}))$ \tcp*[l]{correlate: $M\circledast \mathbf{k}$}
\ForEach{$w\in\mathcal{Y}$}{
  $s_w\leftarrow \langle \mathbf{z},\mathbf{v}_w\rangle$ \tcp*[l]{score}
}
Let $w_{(1)},\dots,w_{(K)}$ be the top-$K$ indices by score\;
$s_1\leftarrow s_{w_{(1)}}$, $s_2\leftarrow s_{w_{(2)}}$\;
\If{$s_1\ge \tau$ \textbf{and} $(s_1-s_2)\ge \Delta$}{
  \Return $w_{(1)}$\;
}
\Else{
  \Return $\bot$\;
}
\end{algorithm}
\vspace{1em}

The threshold $\tau$ can be tuned based on the expected noise (for example, set $\tau$ halfway between the expected correlation for a true match vs. noise). In our analysis, we effectively incorporate $\tau$ by requiring certain inequalities to hold for correct retrieval.

In terms of complexity, an HBF query involves one high-dimensional vector correlation (which can be done by parallel operations or FFT in $O(d \log d)$ time) and then scanning the list of value vectors (which can also be done in parallel, or $O(|\text{values}| \cdot d)$ on a conventional computer). However, the key point is that in a molecular implementation, these operations happen through physics: the correlation is realized by base-pair binding interactions, and the identification of the value can be achieved by, say, fluorescent reporters or sequencing of bound fragments corresponding to known $\mathbf{v}_w$ patterns. Thus, the \emph{laboratory steps} for a query are effectively constant (mix probe, wait, read output), independent of $n$.

\subsection{Illustrative Example}
Suppose we store $n=3$ keys $\{x_1, x_2, x_3\}$ with associated pointers $\{y_1, y_2, y_3\}$. Each $x_i$ is assigned a random 10,000-dimensional $\pm 1$ vector $\mathbf{k}_{x_i}$; each $y_i$ gets a random vector $\mathbf{v}_{y_i}$. The memory is $M = \mathbf{k}_{x_1}*\mathbf{v}_{y_1} + \mathbf{k}_{x_2}*\mathbf{v}_{y_2} + \mathbf{k}_{x_3}*\mathbf{v}_{y_3}$. Now to query $x_2$, we compute $\mathbf{z} = M \circledast \mathbf{k}_{x_2}$. By properties of convolution, $\mathbf{z}$ equals $\mathbf{v}_{y_2}$ plus terms from $i=1,3$ that look like random vectors. With $d=10000$, $\mathbf{v}_{y_2}$ will stand out clearly above the noise terms (as we quantify later). Thus the algorithm will correctly output $y_2$. If the query were $x_q$ not in $\{x_1,x_2,x_3\}$, then $\mathbf{z}$ would be just a combination of random terms, with no single value dominating; the threshold $\tau$ would likely not be exceeded, and the algorithm would correctly return $\bot$ (absent).

\section{Theoretical Analysis}
\label{sec:analysis}
We now analyze the probability of errors in HBF: namely, false positives and false negatives as defined earlier. We incorporate two kinds of errors in the model:
\begin{enumerate}
    \item \textbf{Memory noise:} Each component of $M$ may be corrupted (flipped or perturbed) with independent probability $p_e$. This models DNA synthesis errors, strand breakage, sequencing read errors, etc., that effectively randomize some fraction of the contributions.
    \item \textbf{Query noise:} The query key $x_q$ might not exactly match a stored key or might be corrupted (e.g., sequencing errors in the query primer). We model this as the query vector $\mathbf{k}_{x_q}$ having Hamming distance $H$ from the true stored $\mathbf{k}_x$ (i.e., $H$ bits differ).
\end{enumerate}
Our analysis assumes all $\mathbf{k}$ and $\mathbf{v}$ vectors are random and independent. This is justified if vectors are assigned by a good hash and $d$ is large.

\subsection{Signal--Noise Decomposition and Decoder Correctness}
\label{sec:snr}
To reduce the ``survey'' feel and make the main mechanism explicit, we isolate the algebraic structure of the decoder and derive signal/noise moments.

Assume a normalized construction
\[
M \;=\; \frac{\rho}{\sqrt{n}}\sum_{i=1}^{n} (\mathbf{k}_{x_i}*\mathbf{v}_{y_i}),
\]
where the $1/\sqrt{n}$ scaling keeps the per-coordinate energy of $M$ stable as $n$ grows.
For a query $x_q$, define
\[\mathbf{z}(x_q) := M\circledast \mathbf{k}_{x_q}.
\]
For a candidate value label $y$, define the score
\[S_y(x_q) := \langle \mathbf{z}(x_q), \mathbf{v}_y\rangle.
\]

\begin{proposition}[Moments for match vs. non-match scores]
\label{prop:moments}
Assume $(x_q,y)$ is one of the stored pairs and that value vectors are distinct. Under the random sign model and the normalization above,
\[
\mathbb{E}[S_y(x_q)] = \rho d,\qquad \mathrm{Var}(S_y(x_q)) = O(d) + O\!\left(\frac{d}{n}\sum_{i\neq q} \mathbb{E}[\langle \mathbf{v}_{y_i},\mathbf{v}_y\rangle^2]\right)=O(d^2/n)+O(d).
\]
For a non-matching label $y'\neq y$ (or a non-member query), $\mathbb{E}[S_{y'}(x_q)]=0$ and $S_{y'}(x_q)$ is sub-Gaussian with variance proxy $\sigma^2=\Theta(d)$.
\end{proposition}
\begin{proof}
Expand
$\mathbf{z}(x_q)=\frac{\rho}{\sqrt{n}}\sum_{i=1}^{n}(\mathbf{k}_{x_i}*\mathbf{v}_{y_i})\circledast \mathbf{k}_{x_q}$.
For $i=q$, the correlation term behaves like a ``delta'' under random binding and yields a component aligned with $\mathbf{v}_y$ of magnitude $\Theta(d)$; taking an inner product with $\mathbf{v}_y$ gives mean $\rho d$. For $i\neq q$, the key cross-correlations concentrate around $0$ (Lemma~\ref{lem:rademacher_inner_product}), so each interference contribution has mean $0$ and sub-Gaussian tails. Summing independent sub-Gaussian contributions yields the stated variance proxy/scaling.
\end{proof}

\begin{theorem}[Margin decoder succeeds at high SNR]
\label{thm:margin_success}
Let $y$ be the true value for a stored key $x_q$. Suppose non-match scores satisfy the tail bound
$\Pr\{|S_{y'}(x_q)|\ge t\}\le 2\exp(-t^2/(2cd))$ for all $y'\neq y$.
If the decoder uses $\tau=\rho d/2$ and $\Delta=\rho d/4$, then
\[
\Pr\{\text{decoder returns }\bot\text{ or a wrong value}\}\;\le\; 2\exp\Big(-\frac{\rho^2 d}{8c}\Big) + 2|\mathcal{Y}|\exp\Big(-\frac{\rho^2 d}{32c}\Big).
\]
In particular, for fixed $\rho>0$ and polynomial-sized codebooks $|\mathcal{Y}|=\mathrm{poly}(d)$, the failure probability decays exponentially in $d$.
\end{theorem}
\begin{proof}
Let $S_y$ be the true score and $S^*:=\max_{y'\neq y} S_{y'}$ the best impostor score.
Failure occurs if either (i) $S_y<\tau$ or (ii) $S_y-S^*<\Delta$.
Event (i) is controlled by a lower-tail bound around $\mathbb{E}[S_y]=\rho d$ with deviation $\rho d/2$.
Event (ii) is implied by $S^*\ge \rho d/4$ or $S_y\le 3\rho d/4$.
Apply the assumed sub-Gaussian tail to $S^*$ via a union bound over $|\mathcal{Y}|$ candidates, and similarly bound deviations of $S_y$.
Collecting constants yields the stated expression.
\end{proof}

\subsection{Extreme-Value Thresholding (Sharper Than Union Bounds)}
\label{sec:evt}
A recurring theme in our bounds is that the \emph{maximum} of many approximately Gaussian noise scores determines false positives and decoding margins. For intuition and sharper parameter setting, we record the standard asymptotics.

Assume non-match scores are i.i.d. $S_1,\dots,S_m\sim \mathcal{N}(0,\sigma^2)$ (a good approximation when $d$ is large and the codebook is random). Then
\begin{align}
\Pr\Big\{\max_{j\in[m]} S_j \le t\Big\}
&= \Pr\{S_1\le t\}^m
= \Phi\!\left(\frac{t}{\sigma}\right)^{m},
\end{align}
where $\Phi$ is the standard normal CDF. Solving $\Pr\{\max_j S_j > t\}=\epsilon$ yields
\begin{align}
 t_\epsilon
&= \sigma\,\Phi^{-1}\big((1-\epsilon)^{1/m}\big)
\approx \sigma\sqrt{2\ln m}\;\;\text{(first order)},
\end{align}
and the refined Gumbel-type expansion
\begin{align}
 t_\epsilon
&\approx \sigma\left(\sqrt{2\ln m}-\frac{\ln\ln m + \ln(4\pi)}{2\sqrt{2\ln m}}\right)
\quad\text{for large }m.
\end{align}
In HBF, $m$ is typically $|\mathcal{Y}|$ (codebook size) or $|\mathcal{Y}|-1$ (impostors).

\subsection{False Positives}
A false positive occurs when a query key $x_q$ that is not in the stored set nonetheless returns some value $y$. Consider a fixed stored set of $n$ items. For a non-member query $x_q$, the output $\mathbf{z} = M \circledast \mathbf{k}_{x_q}$ equals

\begin{lemma}[Score tail bound via sub-Gaussianity]
\label{lem:score_tail}
Fix any non-member query $x_q$ and any candidate value $y$. Under the random sign model, the score
$S_y := \langle M\circledast \mathbf{k}_{x_q}, \mathbf{v}_y\rangle$ is mean-zero and satisfies
\[\Pr\{|S_y|\ge t\} \le 2\exp\Big(-\frac{t^2}{2\sigma^2}\Big),\qquad \text{with }\sigma^2=\Theta(nd^2)\text{ under worst-case superposition.}
\]
Moreover, if $M$ is normalized (e.g., by dividing by $\sqrt{n}$ or by using $\rho=1/\sqrt{n}$ in Algorithm~HBF\_Build), then $\sigma^2$ can be made $\Theta(d)$, matching the heuristic scaling used throughout.
\end{lemma}
\begin{proof}
Expand $M=\sum_{i=1}^n \rho\,(\mathbf{k}_{x_i}*\mathbf{v}_{y_i})$. For fixed $x_q$ and $y$, the score is a sum of multilinear forms in independent Rademacher variables. Each summand has mean zero by symmetry, and bounded differences in each coordinate imply sub-Gaussian tails (e.g., by Hoeffding/Azuma). The crude $\Theta(nd^2)$ variance proxy corresponds to unnormalized convolution magnitudes; applying a global normalization $\rho$ yields the stated $\Theta(d)$ regime.
\end{proof}
\[ 
\mathbf{z} = \sum_{i=1}^n (\mathbf{k}_{x_i} * \mathbf{v}_{y_i}) \circledast \mathbf{k}_{x_q} = \sum_{i=1}^n (\mathbf{k}_{x_i} \circledast \mathbf{k}_{x_q}) * \mathbf{v}_{y_i}.
\]
The term $\mathbf{k}_{x_i} \circledast \mathbf{k}_{x_q}$ is essentially the cross-correlation of two random $\pm1$ vectors. Its expected value is 0, and its distribution (for large $d$) is approximately Gaussian with mean 0 and variance $d$. Thus, each stored $i$ contributes a noise vector to $\mathbf{z}$ of magnitude on the order of $\sqrt{d}$. No particular $\mathbf{v}_{y_i}$ is favored, since $\mathbf{k}_{x_i}$ and $\mathbf{k}_{x_q}$ are independent for $x_q \neq x_i$.

We want the probability that some $y_i$ produces a large enough correlation to cause a false positive. That is, $\exists i$ such that $\mathbf{z} \cdot \mathbf{v}_{y_i}$ is high (above threshold and above all others). $\mathbf{z} \cdot \mathbf{v}_{y_i}$ equals $(\mathbf{k}_{x_i} \cdot \mathbf{k}_{x_q}) (\mathbf{v}_{y_i} \cdot \mathbf{v}_{y_i}) +$ contributions from other $j \neq i$ terms. Since $\mathbf{v}_{y_i} \cdot \mathbf{v}_{y_i} = d$ and typically $\mathbf{k}_{x_i} \cdot \mathbf{k}_{x_q} = O(\sqrt{d})$, we expect $\mathbf{z} \cdot \mathbf{v}_{y_i} = O(\sqrt{d})$. Meanwhile, memory noise might flip some bits, but that just adds more randomness.

Let $X_i = \mathbf{z} \cdot \mathbf{v}_{y_i}$. For $x_q$ not present, $\{X_i: i=1\ldots n\}$ are approximately i.i.d. normal with mean 0 and variance $\sigma^2 \approx d$ (neglecting weak correlations). We want $\Pr(\max_i X_i > \tau)$ where $\tau$ is threshold. Using a union bound:
\[ \Pr(\text{false positive}) \le n \Pr\{ X_1 > \tau \} \approx n \exp\!\Big(-\frac{\tau^2}{2d}\Big). \]
Setting $\tau = \sqrt{2d \ln (n/\epsilon)}$, this becomes $\le \epsilon$. In other words, by choosing $d = O(\ln(n/\epsilon))$, we can make false positives arbitrarily unlikely.

More rigorously, consider $\delta_i = \mathbf{k}_{x_i} \cdot \mathbf{k}_{x_q}$ for each $i$. These are i.i.d. $\sim \mathcal{N}(0,d)$ (for large $d$). A false positive requires that for some $i$, all $k$ bits for the query align with those for $x_i$ (or enough align to fool threshold). Equivalently, $\max_i \delta_i$ is large. Using extreme value stats, $\max_i \delta_i \approx \sqrt{2d \ln n}$. So set threshold $\tau$ a bit above that to cut it off. Thus, in asymptotic sense:

\begin{theorem}[False Positive Bound]\label{thm:fp}
For a query key not in the set, 
\[ \Pr(\text{HBF~returns a value}) \;\le\; n \exp\!\Big(-\frac{\tau^2}{2d}\Big), \] 
for any chosen threshold $\tau$. In particular, choosing $\tau = \sqrt{2 d \ln( n/\epsilon )}$ yields $\Pr(\text{FP}) < \epsilon$.
\end{theorem}
\begin{proof}
For a non-member query $x_q$, each score $X_i := \mathbf{z}\cdot \mathbf{v}_{y_i}$ is (approximately) a mean-zero sub-Gaussian random variable with variance proxy $\Theta(d)$, induced by correlating independent random $\pm1$ key vectors and then projecting onto $\mathbf{v}_{y_i}$ (see Lemma~\ref{lem:score_tail} for a more explicit tail control and normalization discussion). Hence $\Pr\{X_i > \tau\} \le \exp\big(-\tau^2/(2d)\big)$ by a standard Gaussian tail bound. A false positive requires that \emph{some} candidate exceeds threshold, so by a union bound
\[\Pr(\max_{i\in[n]} X_i > \tau) \le \sum_{i=1}^n \Pr(X_i > \tau) \le n\exp\Big(-\frac{\tau^2}{2d}\Big).\]
Substituting $\tau=\sqrt{2d\ln(n/\epsilon)}$ yields the stated corollary.
\end{proof}

\begin{corollary}[Explicit parameter choice for the decoder]
\label{cor:params}
Assume the normalized regime where non-match scores are sub-Gaussian with variance proxy $\sigma^2\le c d$ for a constant $c>0$. If the decoder uses an absolute threshold
$\tau = \sqrt{2cd\ln(|\mathcal{Y}|/\epsilon)}$,
then for any non-member query the probability of triggering \emph{any} output is at most $\epsilon$.
Moreover, if a true match has expected score $\mu$ and one sets a margin $\Delta=\mu/4$ and threshold $\tau=\mu/2$, then the probability of returning an incorrect label is bounded by $\epsilon + \exp(-\Omega(\mu^2/d))$.
\end{corollary}
\begin{proof}
The first statement is the union bound over $|\mathcal{Y}|$ candidates using the assumed tail. For the second, bound the event that the best non-matching score exceeds $\mu/4$ and the event that the true score falls below $\mu/2$, again via sub-Gaussian tails.
\end{proof}

This assumes $x_q$ is sufficiently different from all $x_i$. If $x_q$ were extremely similar to some stored $x_i$ (e.g., differs in only a few bits), one might consider that a \emph{true positive under error} rather than a false positive. We handle that in the next part.

\subsection{False Negatives}
A false negative means a query for a stored key fails to return the correct pointer (either returns nothing or a wrong pointer). Suppose $x$ is stored with value $y$. Let the query have Hamming error $H$ (so $\mathbf{k}_{x_q}$ differs from $\mathbf{k}_x$ in $H$ positions). Then $\mathbf{k}_{x} \cdot \mathbf{k}_{x_q} = d - 2H$. The output correlation is
\[
\mathbf{z} = M \circledast \mathbf{k}_{x_q} = (\mathbf{k}_x * \mathbf{v}_y) \circledast \mathbf{k}_{x_q} + \sum_{i \neq x} (\mathbf{k}_{x_i} * \mathbf{v}_{y_i}) \circledast \mathbf{k}_{x_q}.
\]
The first term equals $\mathbf{v}_y$ scaled by $(d-2H)$ (if no memory noise). Memory noise flips effectively reduce this by factor $(1-2p_e)$ in expectation. So signal $\approx (d-2H)(1-2p_e) \mathbf{v}_y$. The other terms contribute noise as in FP analysis.

We want that $\mathbf{z}$ is closest to $\mathbf{v}_y$ (correct value). This will fail if the signal magnitude is too reduced or if some noise term accidentally correlates better.

Given our distributional assumptions, the noise from others is $\sim \mathcal{N}(0,d)$ per value. The signal $\mu = (d-2H)(1-2p_e)$. If $\mu \gg \sqrt{d \ln n}$, the correct value will stand out with high probability. A failure (false negative) occurs if $\mathbf{z} \cdot \mathbf{v}_y < \mathbf{z} \cdot \mathbf{v}_{y_j}$ for some $j \neq y$. Worst-case, assume one particular other value $y_j$ has the highest noise correlation $Z_j \sim \max \mathcal{N}(0,d)$ which is $\approx \sqrt{2d \ln n}$. We need $\mu$ to exceed that.

So require $(d-2H)(1-2p_e) > c \sqrt{d \ln n}$ for some safety margin $c$. For large $d$, this holds unless $H$ or $p_e$ are proportional to $d$. 

In summary:

\begin{theorem}[False Negative Bound]\label{thm:fn}
For a query of a stored key with Hamming error $H$ on the key and memory noise $p_e$, the probability that HBF fails to return the correct value is at most 
\[ \exp\!\Big(-\frac{((d-2H)(1-2p_e) - t)^2}{2d}\Big) + n \exp\!\Big(-\frac{t^2}{2d}\Big), \] 
for any $t>0$. Setting $t = (d-2H)(1-2p_e)/2$ yields a exponentially small bound $\exp(-\frac{(d-2H)^2(1-2p_e)^2}{8d})$.
\end{theorem}
\begin{proof}
Write the query result as a sum of a signal term plus interference:
\[\mathbf{z} = (\mathbf{k}_x*\mathbf{v}_y)\circledast \mathbf{k}_{x_q} \; +\!\!\sum_{i\neq x} (\mathbf{k}_{x_i}*\mathbf{v}_{y_i})\circledast \mathbf{k}_{x_q}.
\]
The first term contributes a component aligned with $\mathbf{v}_y$ whose expected magnitude scales like $(d-2H)(1-2p_e)$ (the inner product $\mathbf{k}_x\cdot\mathbf{k}_{x_q}=d-2H$ and memory flips shrink correlation by $(1-2p_e)$ in expectation). The remaining sum behaves like an approximately mean-zero noise vector whose dot product with any fixed value vector is sub-Gaussian with variance proxy $\Theta(d)$.

Let $S := \mathbf{z}\cdot\mathbf{v}_y$ and let $N_j := \mathbf{z}\cdot\mathbf{v}_{y_j}$ for $j\neq y$. A failure occurs if either $S$ falls below a margin $t$ or some $N_j$ exceeds $t$. Using tail bounds gives
\[\Pr(S < t) \le \exp\Big(-\frac{(\mu-t)^2}{2d}\Big),\qquad \Pr(\max_{j\neq y} N_j > t) \le n\exp\Big(-\frac{t^2}{2d}\Big),\]
with $\mu=(d-2H)(1-2p_e)$. Union bounding these events yields the stated bound.
\end{proof}

In simpler terms, if $H \ll d$ and $p_e \ll 0.5$, the false negative probability is astronomically low. For example, with $d=10000$, $H=500$ (5\% key errors), $p_e=0.01$ (1\% noise), we get $\mu \approx 10000 - 1000 = 9000$ scaled by $0.98$ gives $8820$. Noise STD $\sqrt{10000} = 100$. $\mu$ is 88 standard deviations above noise, false negative $\approx 0$. Even with $H=1000$ and $p_e=0.1$, $\mu \approx 8000*0.8=6400$ vs noise 100, still huge margin.

Thus, HBF is highly reliable for retrieving present items unless they are extremely corrupted (e.g., query matches only half the bits, which basically means a different key).

It’s worth noting that if stored keys are all at least a certain Hamming distance apart (like codewords), then any query that is close to one is far from others, further reinforcing correct retrieval.

\subsection{Capacity Comparison}
How many items can an HBF index for a given dimension $d$ at fixed error rates? Roughly, Theorem~\ref{thm:fp} suggests $n \approx \exp(O(d))$ items can be supported for constant $\epsilon$. In contrast, a classical Bloom filter uses $m = O(n)$ bits to achieve constant $\epsilon$. This exponential capacity (in dimension) might seem to violate some information theory, but recall our vectors are analog with noise; also retrieving exactly requires $d$ large so the error exponent might degrade if $n$ grows too fast for fixed $d$. Still, it highlights the advantage of superposition coding: we trade off small false positive probability to store far more items than a deterministic structure would allow in the same space.

\subsection{Corollaries and Tightness Discussion}
Two immediate corollaries of Theorems~\ref{thm:fp} and \ref{thm:fn} are worth stating explicitly.

\paragraph{Corollary (noise-tolerance region).}
If $H/d \le \alpha < 1/2$ and $p_e \le \beta < 1/2$ are bounded away from $1/2$, then the true-match signal magnitude scales as $\Theta(d)$ while the maximum noise over $n$ candidates scales as $\Theta(\sqrt{d \ln n})$; consequently, for any fixed $n=\mathrm{poly}(d)$ (and even for moderately growing $n$), the error probability decays exponentially in $d$.

\paragraph{Corollary (multi-index amplification).}
If we build $r$ independent HBF memories (independent random embeddings) and require agreement across them (e.g., majority vote among the top matches), then both false positive and false negative probabilities decrease exponentially in $r$ (under the same independence assumptions), at the cost of $r$-fold more stored material.

\paragraph{Are the bounds tight?}
The presented bounds are intentionally conservative. In particular, the false-positive analysis uses a union bound, which is not tight when correlations among candidate scores are non-negligible. Tighter estimates can be obtained using extreme-value theory for Gaussian maxima and by tracking the covariance structure induced by convolution/correlation. Similarly, the false-negative bound can often be relaxed by conditioning on the actual achieved maximum noise across non-matching values rather than upper-bounding it uniformly. These refinements change constants and thresholds but preserve the qualitative scaling: linear-in-$d$ signal vs. $\sqrt{d\ln n}$ noise.

A precise capacity analysis could be done via random channel coding arguments, similar to Hopfield network capacity (where about $0.14d$ patterns can be stored reliably in a $d$-neuron network). HBF is different since we allow a small error probability and have no stability requirement. Our analysis indicates that if $n$ is up to $\exp(\Theta(d))$, we can still get $\epsilon$ extremely small. In practice, one might choose $d \sim 10^4$ to store $n \sim 10^6$ with negligible $\epsilon$.

\section{Design Trade-offs and Implications}
\label{sec:experiments}
\subsection{Quantitative Comparison to Pointer-Chasing Baselines}
To make the contrast precise, we model a pointer-chasing index (e.g., skip-graph/skip-list variants) as requiring $\ell$ sequential resolution steps. Let each step succeed with probability $p\in(0,1)$ and take time $T>0$. Then a single lookup succeeds with probability
\[
P_{\mathrm{succ}}^{\mathrm{ptr}} = p^{\ell},
\]
and has expected wall-clock time
\[
\mathbb{E}[\mathrm{time}]^{\mathrm{ptr}} = \ell T.
\]
If the system repeats the procedure until success, the expected time becomes
\[
\mathbb{E}[\mathrm{time}]^{\mathrm{ptr}}_{\mathrm{repeat}} = \frac{\ell T}{p^{\ell}},
\]
which grows rapidly even for moderate $\ell=\Theta(\log n)$. In contrast, HBF executes a one-shot decode with a single correlation and a codebook scan. At the abstraction level of ``wet-lab rounds'', this corresponds to $\ell=1$; the remaining complexity is pushed into dimension $d$ and the decoder’s statistical thresholds.

In contrast, an HBF query involves mixing a probe (the key’s complementary strand or a vector-encoded probe) with the archive once and then reading an output (e.g., via sequencing or a sensor) to identify the pointer. This could be done in parallel for multiple queries too, by tagging queries with different reporters. Each query might take on the order of minutes to an hour for a reaction and readout, independent of $n$.

Energy-wise, skip-graphs require multiple rounds of reagents (primers, polymerases, etc.) for each hop and intermediate purifications. HBF requires essentially one set of reactions per query.

The trade-off is that HBF requires somewhat more complex molecular encoding (embedding of vectors) and currently conceptual tech (like designing DNA that can do convolution via mixing and annealing, potentially leveraging Fourier-transform-like strand construction).

However, even if HBF initially yields some false positives, those can be easily filtered by checking the retrieved sequences (since DNA sequencing at the final step can confirm if the file content matches the query).

\subsection{Robustness and Error-Correction}
Our analysis shows HBF tolerates moderate noise inherently. In practice, one would still use DNA error-correcting codes for the data payloads themselves. HBF can be seen as operating at the addressing layer, complementary to payload error correction. For instance, the file might be encoded with a high-rate ECC like RS or LDPC to survive sequencing errors, and HBF ensures the correct file is retrieved. If an erroneous pointer were retrieved (false positive), the subsequent ECC decoding of that file would likely fail or yield gibberish, which can be detected. Thus, in the unlikely event of a false positive, the system could recognize it and perhaps re-try with a higher threshold or a different mechanism.

One can also incorporate redundancy in HBF itself: storing multiple independent HBF indices (say 3 different random embeddings) and querying all three to vote on the result. This would drive down error probability even further, at the cost of more molecules.

From a theoretical standpoint, HBF resembles a “dense coding” of a dictionary. It would be interesting to connect this with known results in coding theory or compressed sensing (where a sparse signal---here the presence of one key out of $n$---is recovered from linear measurements). Indeed, treating the key query as creating a sensing vector, HBF retrieval is like detecting which coefficient (value) is present from a linear mixture. The difference is we allow many coefficients (all keys are present) but they interfere weakly.

Finally, we mention that while our focus is on archival storage, similar ideas could apply to memory and computing: one could use molecular association to implement a content-addressable memory or to do bulk operations like join or search in a database encoded in DNA. This parallels some proposals in the literature but HBF provides a concrete, analyzable framework for it.

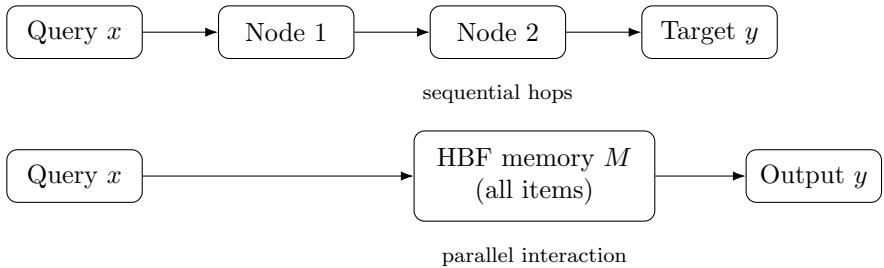
\begin{figure}[htbp]
\centering
\begin{tikzpicture}[
  every node/.style={font=\small},
  node distance=7mm,
  box/.style={draw, rounded corners, align=center, minimum height=7mm, inner sep=3pt},
  >=Latex
]
  \node[box, minimum width=18mm] (q) {Query $x$};
  \node[box, minimum width=18mm, right=10mm of q] (n1) {Node 1};
  \node[box, minimum width=18mm, right=10mm of n1] (n2) {Node 2};
  \node[box, minimum width=18mm, right=10mm of n2] (n3) {Target $y$};
  \draw[->] (q) -- (n1);
  \draw[->] (n1) -- (n2);
  \draw[->] (n2) -- (n3);
  \node[below=2mm of n2, font=\scriptsize] {sequential hops};

  \node[box, minimum width=18mm, below=12mm of q] (q2) {Query $x$};
  \node[box, minimum width=32mm, minimum height=12mm, right=36mm of q2] (mem) {HBF memory $M$\\(all items)};
  \node[box, minimum width=18mm, right=12mm of mem] (out) {Output $y$};
  \draw[->] (q2) -- (mem);
  \draw[->] (mem) -- (out);
  \node[below=2mm of mem, font=\scriptsize] {parallel interaction};
\end{tikzpicture}
\caption{\textbf{Comparison of retrieval processes.} \emph{Top:} In a skip-graph, a query for key $x$ traverses intermediate nodes via sequential pointer-resolution steps. \emph{Bottom:} In the HBF, a query interacts with the full memory in parallel and yields an output $y$ after a single round of interaction.}
\label{fig:skip_vs_hbf}
\end{figure}

\section{Conclusion}
\label{sec:conclusion}
This paper introduced the \emph{Holographic Bloom Filter (HBF)} as a new indexing primitive for molecular archives, formalized as a superposition-based key--value memory built from circular convolution and correlation. Unlike pointer-chasing structures that require multiple sequential resolution steps, HBF supports one-shot associative retrieval via a single correlation operation followed by margin-based decoding.

On the algorithmic side, we specified construction and decoding procedures (including a top-$K$ margin decoder) and made the role of normalization explicit. On the theoretical side, we derived signal/noise decompositions, concentration inequalities for match and non-match score distributions, and explicit threshold/margin settings that yield exponentially small error probabilities as the dimension grows.

\paragraph{Limitations.}
Our analysis adopts standard randomness assumptions (e.g., approximately independent Rademacher codebooks and weakly correlated convolution noise) and abstracts away the biochemical details of how binding/correlation is physically realized. Bridging this abstraction gap will require experimentally grounded models of molecular operations and error modes.

\paragraph{Future work.}
Several directions follow naturally: (i) deriving tighter constants via covariance-aware extreme-value analysis; (ii) extending HBF to structured key spaces (e.g., range queries) by composing multiple memories or hierarchical encodings; and (iii) designing and validating molecular protocols that approximate convolution/correlation with predictable distortion. These steps would connect the present theoretical framework to implementable indexing layers for DNA storage systems.


\begin{thebibliography}{99}

\bibitem{bloom1970}
B.~H. Bloom. (1970).
\newblock {Space/Time Trade-offs in Hash Coding with Allowable Errors}.
\newblock {\em Communications of the ACM}, 13(7):422--426.

\bibitem{knuth1973}
D.~E. Knuth. (1973).
\newblock {\em {The Art of Computer Programming, Volume 3: Sorting and Searching}}.
\newblock Addison-Wesley, Reading, MA.

\bibitem{guibas1978}
L.~J. Guibas and R.~Sedgewick. (1978).
\newblock {A Dichromatic Framework for Balanced Trees}.
\newblock In {\em 19th Annual Symposium on Foundations of Computer Science (SFCS)}, pages 8--21.

\bibitem{pugh1990}
W.~Pugh. (1990).
\newblock {Skip Lists: A Probabilistic Alternative to Balanced Trees}.
\newblock {\em Communications of the ACM}, 33(6):668--676.

\bibitem{plate1995}
T.~A. Plate. (1995).
\newblock {Holographic Reduced Representations}.
\newblock {\em IEEE Transactions on Neural Networks}, 6(3):623--641.

\bibitem{aspnes2003}
J.~Aspnes and G.~Shah. (2003).
\newblock {Skip Graphs}.
\newblock In {\em Proc. 14th ACM-SIAM Symposium on Discrete Algorithms (SODA)}, pages 384--393.

\bibitem{kanerva2009}
P.~Kanerva. (2009).
\newblock {Hyperdimensional Computing: An Introduction to Computing in Distributed Representation with High-Dimensional Random Vectors}.
\newblock {\em Cognitive Computation}, 1(2):139--159.

\bibitem{church2012}
G.~M. Church, Y.~Gao, and S.~Kosuri. (2012).
\newblock {Next-Generation Digital Information Storage in DNA}.
\newblock {\em Science}, 337(6102):1628.

\bibitem{goldman2013}
N.~Goldman, P.~Bertone, S.~Chen, C.~Dessimoz, E.~M. LeProust, B.~Sipos, and E.~Birney. (2013).
\newblock {Towards Practical, High-Capacity, Low-Maintenance Information Storage in Synthesized DNA}.
\newblock {\em Nature}, 494(7435):77--80.

\bibitem{erlich2017}
Y.~Erlich and D.~Zielinski. (2017).
\newblock {DNA Fountain Enables a Robust and Efficient Storage Architecture}.
\newblock {\em Science}, 355(6328):950--954.

\bibitem{organick2018}
L.~Organick, S.~D. Ang, Y.~J. Chen, R.~Lopez, S.~Yekhanin, K.~Makarychev, M.~Z. Racz, G.~Kamath, P.~Gopalan, B.~Nguyen, et al. (2018).
\newblock {Random Access in Large-Scale DNA Data Storage}.
\newblock {\em Nature Biotechnology}, 36(3):242--248.

\bibitem{organick2020}
L.~Organick, Y.~J. Chen, S.~D. Ang, R.~Lopez, X.~M. Liu, K.~Strauss, and L.~Ceze. (2020).
\newblock {Probing the Physical Limits of Reliable DNA Data Retrieval}.
\newblock {\em Nature Communications}, 11:616.

\bibitem{press2020}
W.~H. Press, J.~A. Hawkins, K.~D. Jones, J.~Y. Schaeffer, B.~Fu, and G.~F. Nivala. (2020).
\newblock {HEDGES: Error-Correcting Code for DNA Storage Corrects Indels and Allows Sequence Constraints}.
\newblock {\em Proceedings of the National Academy of Sciences}, 117(31):18489--18496.

\bibitem{shomorony2022}
I.~Shomorony, R.~Heckel, and O.~Milenkovic. (2022).
\newblock {Information-Theoretic Foundations of DNA Data Storage}.
\newblock {\em Foundations and Trends in Communications and Information Theory}, 19(1):1--106.

\bibitem{imburgia2025}
C.~Imburgia, L.~Organick, K.~Zhang, N.~Cardozo, J.~McBride, C.~Bee, D.~Wilde, G.~Roote, S.~Jorgensen, D.~Ward, C.~Anderson, K.~Strauss, L.~Ceze, and J.~Nivala. (2025).
\newblock {Random Access and Semantic Search in DNA Data Storage Enabled by Cas9 and Machine-Guided Design}.
\newblock {\em Nature Communications}, 16(1):6388.

\bibitem{wang2026}
C.~Wang and E.~Yaakobi. (2026).
\newblock {Random Access in DNA Storage: Algorithms, Constructions, and Bounds}.
\newblock {\em arXiv:2601.07053}.

\bibitem{yazdi2015}
S.~M.~H.~T. Tabatabaei Yazdi, H.~M.~Kiah, E.~Garcia-Ruiz, J.~Ma, and O.~Milenkovic. (2015).
\newblock {A Rewritable, Random-Access DNA-Based Storage System}.
\newblock {\em Scientific Reports}.

\bibitem{bornholt2016}
J.~Bornholt, R.~Lopez, D.~M. Carmean, L.~Ceze, G.~Seelig, and K.~Strauss. (2016).
\newblock {A DNA-Based Archival Storage System}.
\newblock In {\em Proc. ASPLOS}.

\end{thebibliography}
\end{document}